\documentclass[11pt,a4paper,twoside]{article}
\usepackage{a4wide, amssymb, amsmath, amsthm, graphics, comment, xspace, enumerate}
\usepackage{graphicx,multirow}
\usepackage[a4paper,colorlinks=true,citecolor=blue,urlcolor=blue,linkcolor=blue
                   , bookmarksopen=true,unicode=true,pdffitwindow=true]{hyperref}
\usepackage[section]{algorithm}
\usepackage[noend]{algpseudocode}
\usepackage{caption}
\usepackage{tabularx}
\usepackage{enumitem, multirow}
\usepackage[usenames,dvipsnames]{color}
\usepackage{alltt}
\usepackage{color}
\usepackage{listings}
\definecolor{string}{rgb}{0.7,0.0,0.0}
\definecolor{comment}{rgb}{0.13,0.54,0.13}
\definecolor{keyword}{rgb}{0.0,0.0,1.0}
\usepackage{amsmath,amsfonts,amssymb,subfigure}
\usepackage{tikz}
\usetikzlibrary{arrows}
\usetikzlibrary{decorations}
\usetikzlibrary{patterns}

\subfiguretopcaptrue
\tikzstyle{vtx}=[circle, inner sep= 0pt, minimum size= 1.2mm, fill]

\hypersetup{pdftitle={Efficient computation of graphs with minimal atom-bond connectivity index}}

\newtheorem{te}{Theorem}[section]
\newtheorem{pro}[te]{Proposition}
\newtheorem{de}{Definition}[section]

\newtheorem{co}[te]{Corollary}

\newtheorem{conjecture}{Conjecture}[section]

\newcommand{\beq}{\begin{eqnarray}}
\newcommand{\eeq}{\end{eqnarray}}

\newcommand{\beqs}{\begin{eqnarray*}}
\newcommand{\eeqs}{\end{eqnarray*}}

\newcommand{\ABC}{{\rm ABC}}

\allowdisplaybreaks

\begin{document}
\title{ {Efficient computation of trees \\with minimal atom-bond connectivity index} }
\date{\vspace{-5ex}}
\maketitle
%\vspace{10mm}
%\vspace{1mm}
\begin{center}
{\large \bf  Darko Dimitrov}
\end{center}
\baselineskip=0.20in
\begin{center}
{\it Institut f\"ur Informatik, Freie Universit\"{a}t Berlin,
\\ Takustra{\ss}e 9, D--14195 Berlin, Germany}
\\E-mail: {\tt darko@mi.fu-berlin.de}
%\end{center}
\end{center}

\baselineskip=0.20in
\vspace{6mm}
\begin{abstract}
The {\em atom-bond connectivity (ABC) index} 
is one of the recently most investigated degree-based molecular structure descriptors, 
that have applications in chemistry.
For a graph $G$, the ABC index  is defined as 
$\sum_{uv\in E(G)}\sqrt{\frac{(d(u) +d(v)-2)}{d(u)d(v)}}$, where $d(u)$ 
is the degree of vertex $u$ in $G$ and $E(G)$ is the set of edges of $G$.
Despite many attempts in the last few years,  it is still an open problem to characterize trees with minimal $ABC$ index.
%To  accomplish that task, an exact computation of graphs with  minimal $ABC$-index can be of enormous help.
In this paper, we present an efficient approach of computing trees with minimal
ABC index,  by considering the degree sequences of trees and 
some known properties of trees with minimal $ABC$ index.
The obtained results disprove some existing conjectures and
suggest new ones to be set.
%, we have computed
%the graphs with minimal $ABC$-index of order up to $3000$. 

%which is an improvement of the currently best computation 
%of graphs of order up to $31$~\cite{fgiv-cstmabci-12}.
\end{abstract}
%
%{\small \hspace{0.25cm} \textbf{Keywords:} First, Secomd, Third Zagreb indices, Uperbound
%
% ------------------------------------------------------------------------------------
%-------------------    Section one:  INTRODUCTION and related results      ----------
% ------------------------------------------------------------------------------------
%
\medskip
\section[Introduction]{Introduction and some related results}
% ----------------------------------------------------------------------------

%Molecular descriptors play a significant role in mathematical chemistry
%and have found  applications in QSPR/QSAR investigations.

Molecular descriptors~\cite{tc-ndc-09} are mathematical quantities that describe the structure or shape of molecules, 
helping to predict the activity and properties of molecules in complex experiments.
Among them, so-called topological indices \cite{db-tird-99} play a significant role.
%Nowadays, there exists a legion of topological indices that
%found some applications in chemistry [4]. 
The topological indices can be classified by the structural properties of graphs used for their calculation.
For example, the Wiener index~\cite{w-rppiams-1948} and the Balaban $J$ index~\cite{b-hddbti-82} are based on the distance of vertices in the respective
graph, the Estrada index \cite{e-c3dms-00} and the energy of a graph~\cite{g-eg-78} 
are based on the spectrum of the graph, the Zagreb group indices~\cite{GT}  and the Randi{\' c} connectivity index~\cite{r-cmb-1975}
depend on the degrees of vertices,  while the Hosoya index \cite{h-ti-1971} is calculated by counting of non-incident
edges in a graph. On the other hand, there is a group of so-called information indices that are based on information functionals~\cite{b-ithiccs-83}. 
More about the information indices and the discriminative power of some established indices,
one can find in~\cite{dgv-iihdpg-12, dk-epge-12, da-hgem-11, fgd-ssdbti-13} and in the works cited therein.

%In view the large number of measures that have been contributed, it is important  tackling this problem is important as it leads to a deeper understand ing of the measures
 
Here, we consider a relatively new  topological index  which attracted a lot of attention last few years.
Namely, in 1998, Estrada et al. \cite{etrg-abc-98} proposed a new vertex-degree-based graph  topological index,
the {\em atom-bond connectivity (ABC) index},
and showed that it can be a valuable predictive tool in the study of the heat of formation in alikeness.
Ten years later Estrada~\cite{e-abceba-08} elaborated a novel quantum-theory-like justification for this topological index.
After that revelation, the interest of ABC-index has grown rapidly.
Additionaly, the physico-chemical applicability of the ABC index was confirmed and extended in several studies
\cite{as-abciic-10,cll-abcbsp-13, dt-cbfgaiabci-10, gg-nwabci-10, gtrm-abcica-12, k-abcibsfc-12, yxc-abcbsp-11}.

Let $G=(V, E)$ be a simple undirected graph of order $n=|V|$ and size $m=|E|$.
For $v \in V(G)$, the degree of $v$, denoted by $d(v)$, is the number of edges incident
to $v$.
Then the atom-bond connectivity index of $G$ is defined as
\beq \label{eqn:001}
\ABC(G)=\sum_{uv\in E(G)}\sqrt{\frac{(d(u) +d(v)-2)}{d(u)d(v)}}.
\eeq

As a new and well motivated graph invariant, the ABC index has attracted a lot of interest in the last several years both in 
mathematical and chemical research communities and numerous results and structural properties of ABC index  
were established~\cite{cg-eabcig-11, cg-abccbg-12, clg-subabcig-12, d-abcig-10, dgf-abci-11, dgf-abci-12, ftvzag-siabcigo-2011, fgv-abcit-09, ghl-srabcig-11, gf-tsabci-12, gfi-ntmabci-12, llgw-pcgctmabci-13, vh-mabcict-2012, xz-etfdsabci-2012, xzd-abcicg-2011, xzd-frabcit-2010}.

The fact that adding an edge in a graph strictly increases its ABC index~\cite{dgf-abci-11} 
(or equivalently that deleting an edge in a graph strictly decreases its ABC index~\cite{cg-eabcig-11})  
has  the following two immediate consequences.

\begin{co}
Among all connected  graphs with $n$ vertices, the complete graph $K_n$ has maximal value of ABC index.
\end{co}

\begin{co}
Among all connected  graphs with $n$ vertices, the graph with minimal ABC index is a tree.
\end{co}

Although it is fairly easy to show that the star graph $S_n$
is a tree with maximal ABC index~\cite{fgv-abcit-09}, despite many attempts in the last years, it is still an open problem
the characterization of trees with minimal ABC index (also refereed as  minimal-ABC trees). 
To accomplish that task,
besides the theoretically proven properties of  the trees with minimal ABC index, computer supported search can be of enormous help.
%Besides many theoretical useful results,
%computer supported search can be of enormous help to understand and prove 
%the structure of  trees with minimal ABC index.
A good example of that is the work done by Furtula et al.~\cite{fgiv-cstmabci-12}, 
where the trees with minimal ABC index of up to size of $31$ were computed.
There, a brute-force approach of generating all trees of a given order, 
speeded up by using  a distributed computing platform, was applied.
%The computation was speeded up by using  a distributed computing platform.
%There, a brute-force approach of generating all trees of a given order was applied.
%The computation was speed up by using  a distributed computing platform.

Here, we improve the computer search in two ways. Firstly, we consider only the degree sequences
of trees. We would like to stress that the number of degree sequences of a given length $n$ is significantly smaller than
the number of all trees of order $n$. For example, the number of trees with $32$ vertices is $109 \, 972 \, 410\, 221$ \cite{Sloane-no-trees},
while the number of degree sequences of length $32$ is $5604$ (see Table~\ref{t-enumDS-1}).
Secondly, to speed up the computation, we generate only degree sequences of trees that correspond to 
some known structural properties of the trees with minimal ABC index 
(Propositions~\ref{pro:20}, \ref{pro:30}  and \ref{pro:40} from Section~\ref{trees-with-minABC} ). 
Thus, using a single PC, 
we have identified all trees with minimal ABC of order up to $300$.
The obtained results strengthen the believe that some conjectures are true and reject other conjectures.

In the sequel, we present some additional results and notation that will be used in the rest of the paper.
A vertex of degree one is a {\it pendant vertex}.
As in \cite{gfi-ntmabci-12}, a sequence of vertices of a graph $G$, $S_k=v_0 \, v_1 \dots v_k$,  will be  called a {\it pendant path} if 
each two consecutive vertices in $S_k$ are adjacent in $G$, $d(v_0)>2$, $d(v_i)=2$, for $i=1, \dots, k-1$, and $d(v_k)=1$.
The length of the pendant path $S_k$ is $k$.
A sequence $D=(d_1, d_2, \dots, d_n)$ is {\it graphical} if there is a graph
 whose vertex degrees are $d_i$, $i=1,\dots,n$. If in addition
 $d_1 \geq d_2\geq \dots \geq d_n$, then  $D$ is a
 {\it degree sequence}.
Let ${\bf D_n}$ be the set of all degree sequences of trees of length $n$.

In~\cite{w-etwgdsri-2008} Wang defined a {\em greedy tree} as follows.

\begin{de}[\cite{w-etwgdsri-2008}]\label{def-GT}
Suppose the degrees of the non-leaf vertices are given, the greedy tree is achieved by the following `greedy algorithm':
\begin{enumerate}
\item Label the vertex with the largest degree as $v$ (the root).
\item Label the neighbors of $v$ as $v_1, v_2,\dots,$ assign the largest degree available to them such that $d(v_1) \geq d(v_2) \geq \dots$
\item Label the neighbors of $v_1$ (except $v$) as $v_{11}, v_{12}, \dots$ such that they take all the largest
degrees available and that $d(v_{11}) \geq d(v_{12}) \geq . . .$ then do the same for $v_2, v_3,\dots$
\item Repeat 3. for all newly labeled vertices, always starting with the neighbors of the labeled vertex with largest whose neighbors are not labeled yet.
\end{enumerate}
\end{de}
\noindent
%Further, we assume that the root vertex $v$ is at level $0$ of the tree. The vertices at level $i$ are at distance $i$ to $v$.

\noindent
The following result  by Gan, Liu and You~\cite{gly-abctgds-12} characterizes the trees with minimal ABC index with prescribed degree sequences. 

\begin{te}\label{thm-DS}
Given the degree sequence, the greedy tree minimizes the ABC index.
\end{te}

\noindent
The same result as in Theorem~\ref{thm-DS}, using slightly different notation and approach, was obtained by
Xing and Zhou \cite{xz-etfdsabci-2012}.
Since the  Theorem~\ref{thm-DS} plays a crucial role in our computation, the first important issue
is how  to enumerate efficiently degree sequences of trees. This problem is considered in the next section. 
%Next we consider the problem of enumerating degree sequences of trees. 

\section[Enumerating degree sequences of trees]{Enumerating degree sequences of trees}\label{sec:DS}
% ----------------------------------------------------------------------------

There exist several algorithms for enumerating degree sequences of graphs. 
A comprehensive source of references of such algorithms can be found in \cite{ilms-eghha-11}.
Clearly, each of those algorithms can be used for  enumerating degree sequences of trees just by considering only the degree sequences
with sum of degrees equals to $2n - 2$, where $n$ is the length of the degree sequences. However, this is not an efficient approach, 
because most of the generated degree sequences are not degree sequences of trees.
For an illustration,  the number of all degree sequences of length $29$ is $2 \,022 \,337 \,118 \,015 \,338$ \cite{Sloane-dsequnce},  
while the number of degree sequences that correspond to trees of order $29$ is $3010$ (see Table~\ref{t-enumDS-1}).
Thus, it is not a surprise that the largest reported enumerated degree sequences of graphs was only of length $29$, 
with running time of $6733$ days, distributed to 200 PCs containing about 700 cores \cite{ilms-eghha-11}.
%needed of about 
%which 
%the running time 
%needed  for enumerating all degree sequences of graphs of order $29$ 
%(which is the best known reported result \cite{Sloane-dsequnce}) was $6733$ days, distributed 
 %to about 200 PCs containing about 700 cores
%\cite{ilms-eghha-11}.
%For an illustration, the running time 
%needed  for enumerating all degree sequences of graphs of order $29$ 
%(which is the best known reported result \cite{Sloane-dsequnce}) was $6733$ days, distributed 
 %to about 200 PCs containing about 700 cores
%\cite{ilms-eghha-11}.
Since we are not aware of an algorithm specialized only for  enumerating degree sequences of trees,
we present  such an algorithm in this section.
Our algorithm is related to the algorithm of enumerating degree sequences of graphs presented by 
Ruskey et al \cite{Ruskey94alleycats}, and exploit the so called ``reverse search'', a term originated by Avis and Fukuda \cite{af-rse-96}. 
Therefore, in the sequel, we will adopt the notation used in \cite{Ruskey94alleycats}.
The main result, on which our algorithm is based, is the following characterization of a degree sequence of a tree.

\begin{te} \label{tree-DS-characterization}
A sequence of integers $D=(d_1, d_2, \cdots, d_n)$, 
with $n-1 \geq d_1 \geq d_2 \geq \cdots \geq d_{n-m} > d_{n-m+1}= \cdots =d_n=1$, 
is the degree sequence of a tree if and only if  $C=(c_1, c_2,  \cdots, c_{n-d_{n-m}+1})$ 
is the degree sequence of a tree, where
\beq \label{eq-10}
c_i = \left\{ \begin{array}{ll}
        d_i      & \quad i \leq n-m-1;\\
       1 & \quad otherwise.\end{array} \right.  
\eeq
\end{te}
\begin{proof}
Let $T_C$ be a tree with degree sequence $C=(c_1, c_2,  \cdots, c_{n-d_{n-m}+1})$, with 
$c_1 \geq c_2 \geq \cdots \geq c_{n-m-1} > c_{n-m}= \cdots =c_{n-d_{n-m}+1}=1$ and $c_{n-m-1} \geq d_{n-m} \geq 2$,
satisfying~(\ref{eq-10}).

To prove the easier direction of the equivalence,
%assume that $C=(c_1, c_2,  \cdots, c_{n-d_{n-m}+1})$, with 
%$c_1 \geq c_2 \geq \cdots \geq c_{n-m-1} > c_{n-m}= \cdots =c_{n-d_{n-m}+1}=1$,
%is a degree sequence of a tree $T_C$.
just add $d_{n-m}-1$ pendant vertices to a pendant vertex of $T_C$,
obtaining a tree $T_D$. The degree sequence that corresponds to $T_D$ is
$D=(d_1, d_2, \cdots, d_n)$, with $n-1 \geq d_1 \geq d_2 \geq \cdots \geq d_{n-m} > d_{n-m+1}= \cdots =d_n=1$.

The other direction of the equivalence, we prove as follows.
Let $D=(d_1, d_2, \cdots, d_n)$, with $n-1 \geq d_1 \geq d_2 \geq \cdots \geq d_{n-m} > d_{n-m+1}= \cdots =d_n=1$,
be a degree sequence of a tree  $T_D$. Let $v_{n-m}$ be the vertex with degree $d_{n-m}$.
If  $v_{n-m}$ has $d_{n-m}-1$ pendant vertices, then delete them obtaining the tree $T_C$.
If this is not a case, i.e., $v_{n-m}$ has $d >1$ adjacent vertices of degree bigger than one that comprised a 
set $U=\{ u_1, u_2,\dots,  u_d \}$.
Let $U_1$ be a set of adjacent vertices to $u_1$. 
First, delete all edges between $u_1$ and  vertices in $U_1\setminus \{v_{n-m}\}$ and 
add edges between vertices in $U_1\setminus \{v_{n-m}\}$ and a pendant vertex whose distance to $u_1$ is bigger
than its distance to any other vertex in $U$.
Notice that $T_D$ has more than $d_{n-m}$ pendant vertices, therefore such pendant vertex must exists.
Repeat the same as for $u_1$, for the rest of the vertices $ u_2, u_3, \dots,  u_d$, considering one vertex per step until $v_{n-m}$ has $d_{n-m}-1$ pendant vertices,
 obtaining a tree $T_D'$. 
Observe that $T_D'$ has the same degree sequence as $T_D$.
%, and $v_{n-m}$ in $T_D'$ has $d_{n-m}-1$ pendant vertices.
Finally, in $T_D'$ delete all  $d_{n-m}-1$ pendant vertices adjacent to $v_{n-m}$, obtaining the tree $T_C$.
% with degree sequence 
%$C=(c_1, c_2,  \cdots, c_{n-d_{n-m}+1})$ with $c_1 \geq c_2 \geq \cdots \geq c_{n-m-1} > c_{n-m}= \cdots =c_{n-d_{n-m}+1}=1$.

\end{proof}
\noindent
Let $\bf S_i$ be the set of all sequences $D_i=(d_1, d_2, \cdots, d_i)$, with fixed length $i$ , where $1<i \leq n$ and
$n-1 \geq d_1 \geq d_2 \geq \cdots \geq d_{h_i} > d_{h_i+1}= \cdots =d_i=1$.
Notice that, $d_{h_i}$ denotes the smallest degree in $D_i$ larger than one.
Define a function $f_i: {\bf S_i} \times d_{h_i} \rightarrow {\bf S_{i-d_{h_i} +1}} \times d_{h_i-1} $ such that for  a given $D_i \in \bf S_i$,
 and $C= (c_1, c_2, \cdots, c_{i-c_{h_i} +1})$, it holds that   $(C, c_{h_i-1})=f_i(D_i, d_{h_i})$ if
$$
c_k = \left\{ \begin{array}{ll}
        d_k      & \quad k \leq {h_i} -1;\\
       1 & \quad otherwise.\end{array} \right. 
$$

By Theorem~\ref{tree-DS-characterization} and definition of the function $f_i$,
we have the following two corollaries.

\begin{co}
For $i>0$ and $D_i \in {\bf S_i}$, the sequence $D_i \in {\bf D_i}$
if and only if $f_i(D_i, d_{h_i})=(D_{i-d_{h_i}+1}, d_{h_i-1}) \in {\bf D_{i-d_{h_i}+1}}$.
\end{co}

\begin{co} \label{tree-DS-characterization-2}
Let $C=(c_1, c_2, \dots, c_{h_i}\dots, c_{i-z}, c_{i-z+1})  \in  \bf D_{i-z +1}$, with $c_{h_i}$ the smallest degree bigger than 1, and
$2\leq z  \leq c_{h_i}$.
The sequence $D_i=(d_1, d_2, \dots, d_{i}) \in f_i^{-1}(C, c_{h_i})$ if and only if 
%$f_i(D_i, d_{h_i})=(D_{i-d_{h_i}+1}, d_{h_i-1})$
$$
d_k = \left\{ \begin{array}{ll}
        c_k      & \quad k \leq h_i;\\
        z   & \quad k = h_i+1; \\
       1 & \quad otherwise.\end{array} \right. 
$$
\end{co}
%\begin{proof}
%The proof follows from the definition of $f_i$.
%\end{proof}
%
\noindent
The following example illustrates Corollary~\ref{tree-DS-characterization-2}:
$$
f^{-1}(65111111111) \supseteq \{655111111111111, 65411111111111, 6531111111111, 652111111111\}.
$$

\bigskip
\noindent
One may straightforwardly implement Corollary~\ref{tree-DS-characterization-2} using recursion
to enumerate degree sequences of trees.
Our C++ implementation 
%A straightforward implementation in C++ of the algorithm that follows from Lemma~\ref{tree-DS-characterization-2} is given in the appendix.
was run on $2.3$ GHz Intel Core i$5$ processor with $4$GB 1333 MHz DDR3 RAM.
The performance of the algorithm is presented in Table~\ref{t-enumDS-1}. 

%For $3< i<n$ (where $n$ is the length of the required degree sequence): 
%\begin{itemize}	
%\item Create the degree sequence of the star $S_i$ 
%\item Let d be the smalls degree in the sequence larger than $1$
%\item For  $2 \leq k \leq d$ :	
%	\subitem- add a degree  $k$ after d in the degree sequence
%	\subitem- add $k-2$ '$1$s' in the degree sequence
 %        \subitem- stop if the length of the sequence is $n$ (or over $n$)
%\end{itemize}
%Show the correctness of both algorithms (as lemmas): both generate valid degree sequences.

\begin{table}[width=\columnwidth, http]
\caption{Performance of the algorithm for enumerating degree sequences of trees.
For degree sequences of length $n$,
$S(n)$ denotes the number of degree sequences, $T(n)$ denotes the total running time
and $S(n)/T(n)$ denotes the amortized running time for generating a sequence.}
\begin{center}
\begin{tabular}{||r|r|r|r||r|r|r|r||}
%\hline
%\multicolumn{6}{|c|}{Igea} \\
\hline
\hline
$n$  & $S(n)$ & $T(n)$  & $T(n)/S(n)$[ms] & $n$  & $S(n)$ & $T(n)$ & $T(n)/S(n)$[ms] \\
\hline
\hline
21  &    490   &   3.031ms  &  0.00618571 &  34  &     8349  & 0.041s   & 0.00502455\\
\hline
22  &   627    &   3.989ms & 0.00636204 &  35  &       10143 &   0.051s  & 0.00501538\\
\hline
23  &    792   &  4.875ms   & 0.00615530 & 40&  26015& 0.098s & 0.00378359\\
\hline
24  &    1002   &  6.082ms     & 0.00606986 &50  & 147273 & 0.612s & 0.00415907\\
\hline
25  &    1255   &  7.884ms   & 0.00628207 &  60&  715220 & 3.230s & 0.00451627 \\
\hline
26  &    1575   &   9.996ms  & 0.00634667 & 70& 3087735 & 15.300s & 0.00495518 \\
\hline
27  &    1958   &  13.083ms   & 0.00668182 & 80&  12132164 & 1m3s & 0.00520250 \\
\hline
28  &     2436  &  14.434ms  & 0.00592529 & 90&  44108109& 4m5s & 0.00556759  \\
\hline
29  &     3010  &   20.086ms  & 0.00667309 & 100 & 150198136 &  14m27s& 0.00577865 \\
\hline
30  &    3718   &  18.821ms   & 0.00506213 & 110 & 483502844 & 51m26s & 0.00635496 \\
\hline
31  &    4565   &  23.031ms & 0.00504513  & 120& 1482074143 & 2h39m8s &  0.00647126\\
\hline
32  &       5604  &   29.523ms  & 0.00526820 &  130&  4351078600 & 7h36m43s & 0.00629813\\
\hline
33  &       6842 &  33.430ms    & 0.00488600 & 140&  12292341831& 21h7m44s &0.00618793 \\
\hline
\hline
\end{tabular}
\end{center}
\label{t-enumDS-1}
\end{table}

The time needed for generating a degree sequence indicates that algorithm runs in 
constant amortized time, so the total running time of the whole program is $O(S(n))$.
Similarly as in the case of general graphs~\cite{Ruskey94alleycats}, it may be 
rather difficult to determine the precise complexity of the algorithm as function of $n$, 
which remains an open problem.

We have tested and analyzed our enumeration algorithm of  degree sequences up to $n=140$,
for which the algorithm  ran about 21 hours. On our single processor platform,
 we would expect that for enumerating degree sequences of length $160$ will run about $7$ days. 
One can improve 
the computation by considering the distribution of the graphical sequences according 
to their first element. That can help to design an algorithm that computes 
the new values of $S(n)$ by  slicing of the computations belonging to a given value of $n$,
similarly as it was done in~\cite{ilms-eghha-11}.

Having all degree sequences of a particular length, in the next section we proceed to determine
the trees with minimal ABC index. 
%Compare the maximal length of DS for trees and general graphs (Sage, the people form Hungary).

%\bigskip
%As alternative we can use an algorithm based on partitions. Cite Nathann, sage.

%How to generate only partitions with at most one 'component' with value  2 and no 'components' with value 1.
%We fix the numbers of "$1$s", denote by $m$, and look for a partition of $2n-2-m$ in $n-m$ numbers $n-1 \geq d_1 \geq d_2 \cdots d_{n-m} \geq 2$. It is easy to show that $d_1 \leq m$. If $d_1=m$, then $d_2=d_3 =\cdots =d_{n-m}=2$.

%Let $n_1$ be the number of vertices of degree one. Clearly,  $n_1=1$ if and only if $n=2$ ($n$ is the length of a degree sequence).

%Connection with Euler generating function.

%

%\section[Additional results used in the computation]{Additional results used in the computation} \label{general}
\section[Trees with minimal atom-bond connectivity index]{Trees with minimal atom-bond connectivity index} \label{trees-with-minABC}

Our algorithm  of identifying the trees with minimal ABC index is comprised of the following steps:

\begin{enumerate}
\item Enumerate all degree sequences as described in Section~\ref{sec:DS}.
\item Find corresponding `greedy trees' for each generated degree sequence applying Theorem~\ref{thm-DS}.
\item Calculate the ABC index of each `greedy tree' and select the tree with minimal value.
\end{enumerate}

Computationally, step $1.$ is the most expensive one, and as pointed in Section~\ref{sec:DS}, it can be parallelized to some extend
to run in distributed framework.
Steps $2.$ and $3.$ can be implemented  straightforwardly and there computational cost is linear with the respect to the length of 
the degree sequence.
Although this approach is computationally superior to the computer assisted search presented in~\cite{fgiv-cstmabci-12},
it can be considerably improved by enumerating only the degree sequences that satisfy
the following structural properties of the minimal ABC trees.  

\begin{pro}[\cite{gfi-ntmabci-12}] \label{pro:20}
If $n \geq 10$, then the $n$-vertex tree with minimal ABC index does not contain pendent paths of length $k  \geq 4$.
\end{pro}

\begin{pro}[\cite{gfi-ntmabci-12}] \label{pro:30}
If $n \geq 10$, then the $n$-vertex tree with minimal ABC index contains at most one pendent path of length $k =3$.
\end{pro}

%The following result, which is a affirmative proof of a conjecture by Gutman et al.~\cite{gfi-ntmabci-12}, will be also use 
%to sped up the computer
%search for trees with minimal ABC index.

\begin{pro}[\cite{llgw-pcgctmabci-13}] \label{pro:40}
If $n \geq 10$, then each pendent vertex of the $n$-vertex tree $G$ with minimal ABC index belongs to a pendent path of length $k , \;2 \leq k \leq 3$.
\end{pro}

% we do not need the following result. it follows from 'greedy tree property'
%
%\begin{pro}[\cite{gfi-ntmabci-12}] \label{pro:10}
%If $n \geq 10$, then the $n$-vertex tree with minimal ABC index does not contain internal paths of length $k  \geq 2$.
%\end{pro}

%\begin{proof}
%Since $G$ has minimal ABC index, by Theorem~\ref{thm-DS} $G$ must be a greedy tree.
%A pendant vertex can belong only to the last or before the last level of a greedy tree.

%Now, we assume that  $G$ has a pendant vertex $x$  that does not belong to a pendent path of length $k , 2 \leq k \leq 3$,
%i.e., the parent vertex of $x$ is vertex of degree at least $3$. Denote by $y$ the parent vertex of $x$.
%\end{proof}

\noindent
Considering all these results, 
that reduce significantly the number of degree sequences, we have implemented an algorithm that identifies trees with minimal ABC index.
On our single processor platform, we have calculated all trees with minimal ABC index of order up to $300$ in about 15 days,
which is a significant improvement over the previous similar work~\cite{fgiv-cstmabci-12}, where a grid infrastructure of about 400 CPUs 
was used to identify all trees with minimal ABC index of order up to $31$.
%On our single processor platform, we have calculated all trees with minimal ABC index of order up to $300$ in about 15 days,
%which is a significant improvement over the previous similar work~\cite{fgiv-cstmabci-12}, where using a grid infrastructure of about 400 CPUs 
%all trees with minimal ABC index of order up to $31$ were identified.
All obtained trees with minimal  ABC index are summarized in Figures~\ref{fig-conjecture-1}~and~\ref{fig-conjecture-2}. 
For the sake of completeness, we include also the results for $7 \leq n \leq 31$, which were already  obtained in \cite{fgiv-cstmabci-12, gfi-ntmabci-12}.
For $n \leq 6$, the minimal ABC trees are paths $P_n$ and they are omitted in the figures.

\begin{figure}[http]
\begin{center}
%\vspace{-0.3cm}
\includegraphics[scale=0.77]{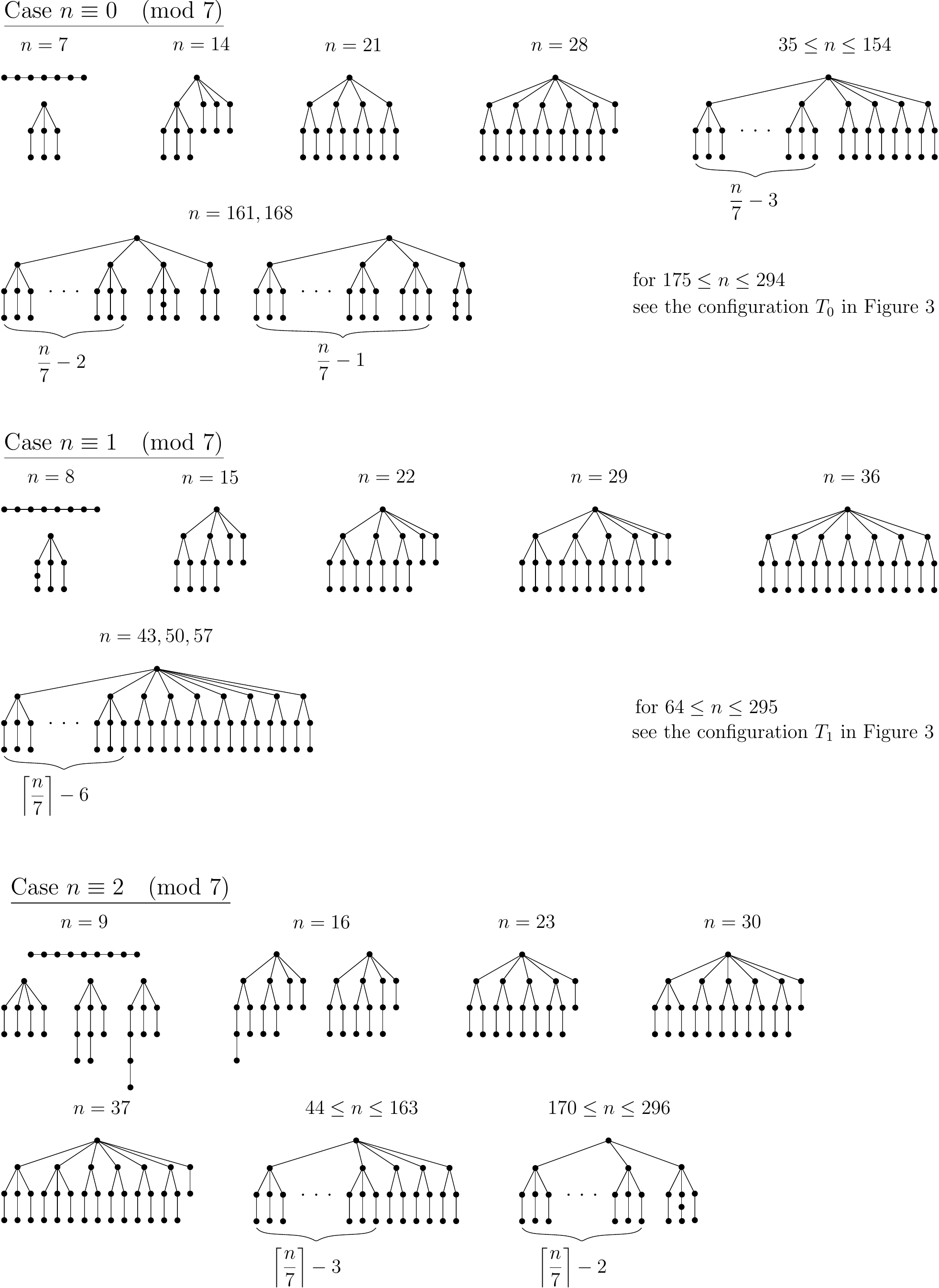}% [height=3.5cm,width=15cm]
\caption{Trees of order $n$, $7 \leq n \leq 300$, with minimal ABC index obtained by computer search - 
cases $n \equiv 0,1,2 \pmod{7}$.}
%\label{Unicyclic-Max-M3}
\label{fig-conjecture-1}
%\vspace{-0.3cm}
\end{center}
\end{figure}

\begin{figure}[http]
\begin{center}
%\vspace{-0.3cm}
\includegraphics[scale=0.77]{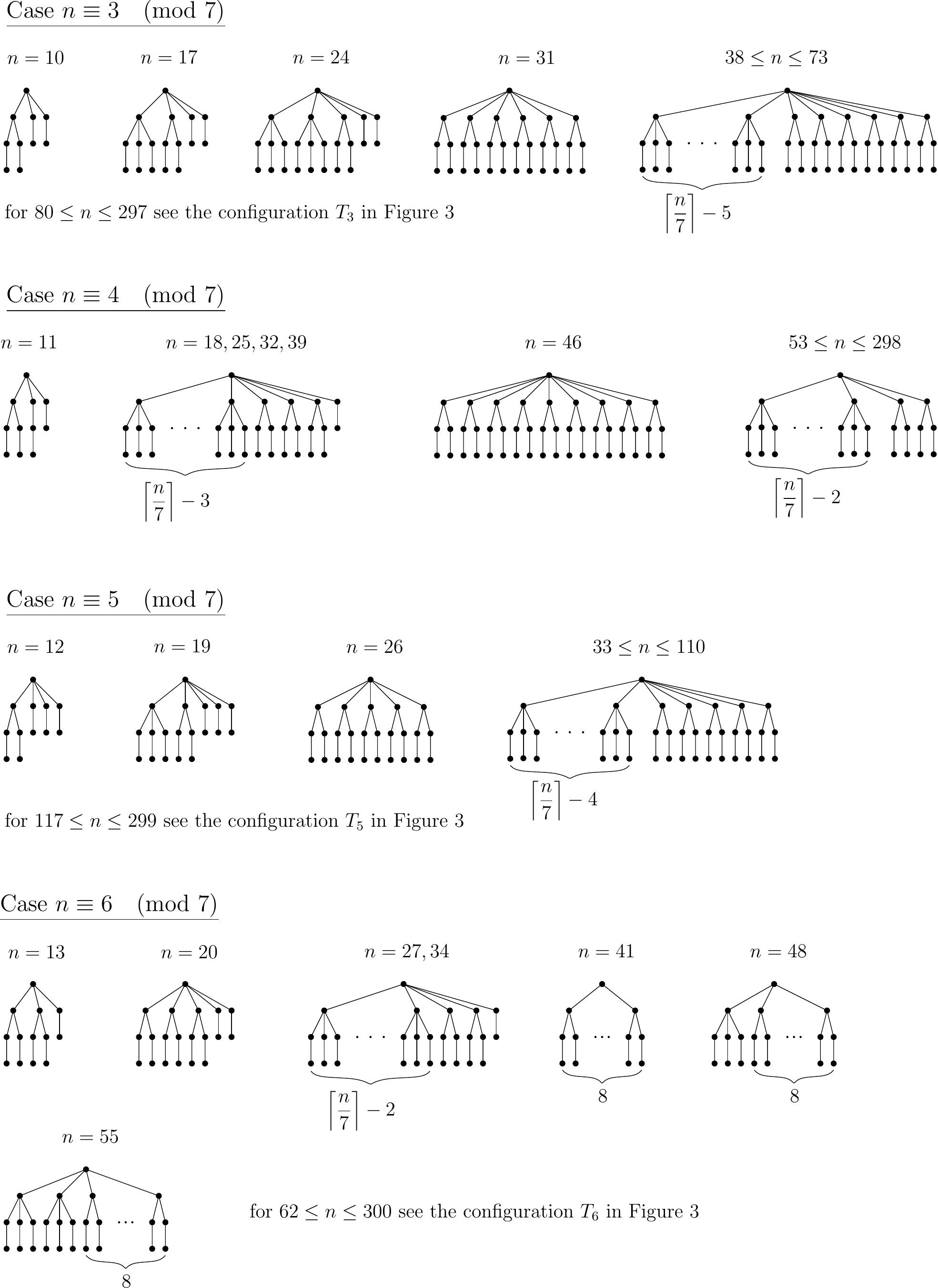}% [height=3.5cm,width=15cm]
\caption{Trees of order $n$, $7 \leq n \leq 300$, with minimal ABC index obtained by computer search - 
cases $n \equiv 3, 4, 5, 6 \pmod{7}$.}
%\label{Unicyclic-Max-M3}
\label{fig-conjecture-2}
%\vspace{-0.3cm}
\end{center}
\end{figure}

From the examples in Figures~\ref{fig-conjecture-1}~and~\ref{fig-conjecture-2}, 
the following observation can be made:
\smallskip
\noindent
If $\Delta_1$  is a maximal degree of a minimal ABC-tree with $n$ vertices, 
and $\Delta_2$  is a maximal degree of a minimal ABC-tree with $n+1$ vertices,
 then $\Delta_2 =   \Delta_1 \pm 1$. 
\smallskip
\noindent
Note that this observation is not necessarily true
for larger minimal ABC-trees. However, applying the above restriction, one can significantly speed up the calculations.
In this way we were able to obtain results that lead to disprovement of some existing conjectures.
The plausible structural computational model  and its refined version in \cite{fgiv-cstmabci-12}, 
is based on the main assumption that  the minimal ABC tree posses a single {\it central vertex},
or said with other words, it is based  on the assumption that 
the vertices of a minimal ABC tree of degree $\geq 3$ induce a star graph.
The configuration $T_4$ in Figure~\ref{fig-conjecture}, for $n \equiv 4 \pmod{7}$ and $n \geq 312$,
is an counterexample to that assumption.

Applying the above constrain on the maximal degree of a minimal ABC-tree,
we have also obtained another counterexample.
Namely, the configuration $T_2$ in Figure~\ref{fig-conjecture} indicates different structure of the
minimal ABC trees, in the case when $n \equiv 2 \pmod{7}$ and $n \geq 1185$, than the
structure suggested for this case in \cite{fgiv-cstmabci-12, gf-tsabci-12}.

In meantime, independently to this work,
the above two counterexamples  were discovered and recently published.
%there were recently published few papers whose results partialy overlap with the results obtained here.
Namely, the counterexample $T_4$ in Figure~\ref{fig-conjecture} was published in \cite{ahs-tmabci-13},
and the counterexample $T_2$ in Figure~\ref{fig-conjecture} was published in \cite{ahz-ltmabci-13}.
These two counterexamples, as well as the most already obtained results on this topic, 
will appear in a review article  \cite{gfahsz-abcic-2013}.

As a consequence to these counterexamples, we present a revised version of  the conjecture 
by Gutman and Furtula \cite{gf-tsabci-12}  about the trees with minimal ABC index.
Before we state it, it is worth to mention that the original versions of the conjecture by Gutman and Furtula, 
with slightly corrections, but still supporting the idea of existence of a central vertex, 
 was shown to be true for the so-called {\em Kragujevac trees}~\cite{hag-ktmabci-14}.

%Inspired by the idea of the central vertex, a slightly corrected version of the conjecture 
%by Gutman and Furtula \cite{gf-tsabci-12} was shown to be true for the so-called 
%{\em Kragujevac trees}~\cite{hag-ktmabci-14}.

%Our computations showed that the new conjecture is true when the order of a tree is up to $2000$.
%Our computations support the new conjecture in the cases when the order of a tree is up to $2000$.

\begin{conjecture} \label{conj-mod7}
Let $G$ be a tree with minimal ABC index among all trees of size $n$.
\begin{itemize}
\item[(i)] If $n \equiv 0 \pmod{7}$ and $n \geq 175$, then $G$ has the structure $T_0$ depicted in Figure~\ref{fig-conjecture}.
\item[(ii)] If $n \equiv 1 \pmod{7}$ and $n \geq 64$, then $G$ has the structure $T_1$ depicted in Figure~\ref{fig-conjecture}.
\item[(iii)] If $n \equiv 2 \pmod{7}$ and $n \geq 1185$, then $G$ has the structure $T_2$ depicted in Figure~\ref{fig-conjecture}.
\item[(iv)] If $n \equiv 3 \pmod{7}$ and $n \geq 80$, then $G$ has the structure $T_3$ depicted in Figure~\ref{fig-conjecture}.
\item[(v)]  If $n \equiv 4 \pmod{7}$ and $n \geq 312$, then $G$ has the structure $T_4$ depicted in Figure~\ref{fig-conjecture}.
\item[(vi)]  If $n \equiv 5 \pmod{7}$ and $n \geq 117$, then $G$ has the structure $T_5$ depicted in Figure~\ref{fig-conjecture}.
\item[(vii)]  If $n \equiv 6 \pmod{7}$ and $n \geq 62$, then $G$ has the structure $T_6$ depicted in Figure~\ref{fig-conjecture}.
\end{itemize}
\end{conjecture}

%We would like to note that the original conjecture  \cite{gf-tsabci-12} in slightly corrected version  
%was shown to be true for the so-called  {\em Kragujevac trees}   (tress with a central vertex
%and regular branching structure) \cite{hag-ktmabci-14}.

\begin{figure}[h]
\begin{center}
%\vspace{-0.3cm}
\includegraphics[scale=0.750]{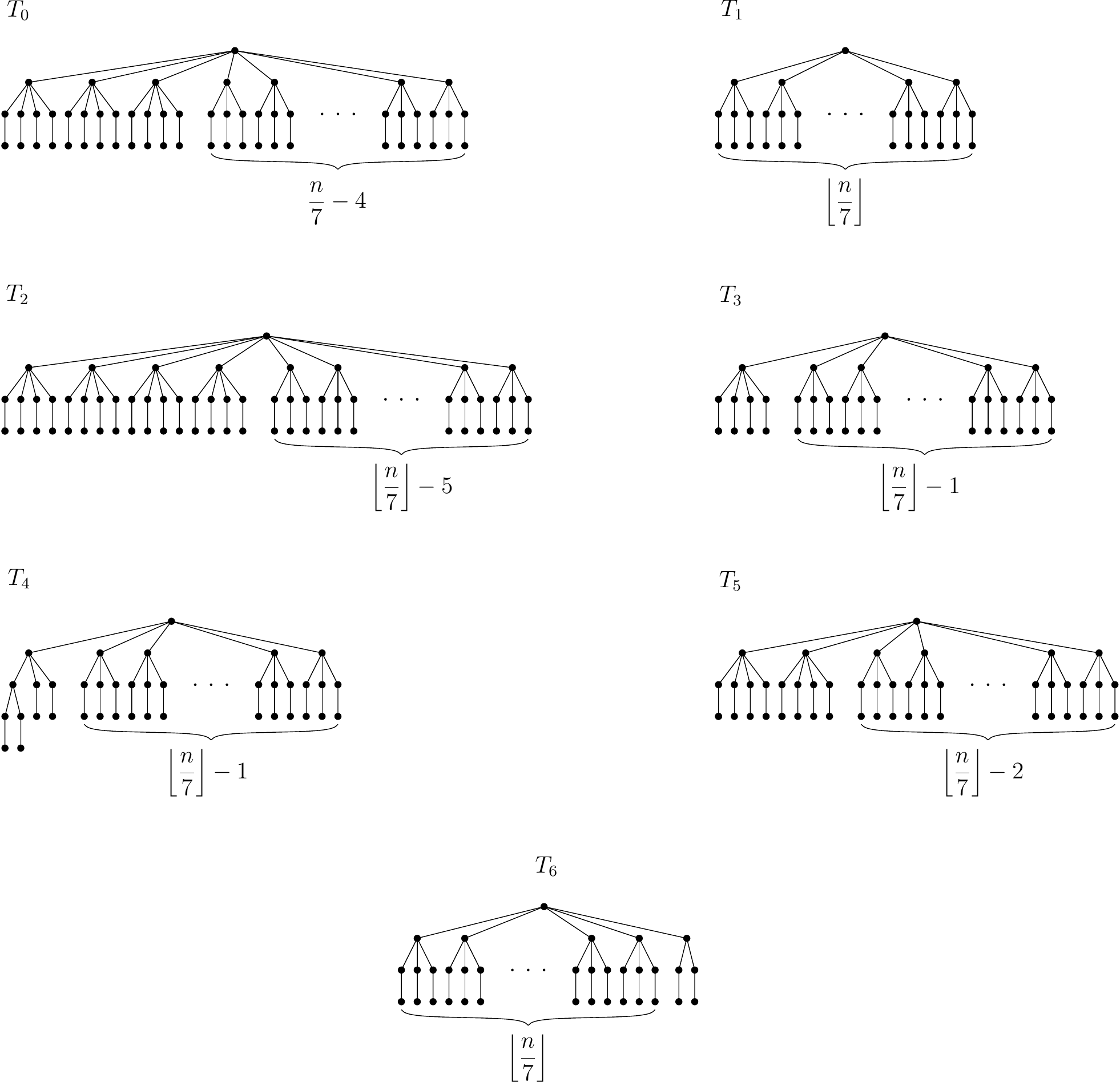}% [height=3.5cm,width=15cm]
\caption{Types of trees with minimal ABC index that correspond to Conjecture~ \ref{conj-mod7}.}
%\label{Unicyclic-Max-M3}
\label{fig-conjecture}
%\vspace{-0.3cm}
\end{center}
\end{figure}
%
%\bigskip
\noindent
The obtained computational results also indicate the following extension of the Proposition~\ref{pro:30}.
\begin{conjecture} \label{conj-no-pendant-paths}
A minimal ABC tree of order $n >1178$ does not contain a pendant path of length three.
\end{conjecture}

We would like to note that our computations show only minor 
violation of the assumption about the central vertex.
However, to determine  how big this violation is, is still an open problem.
The computations here also strengthen the already existing believe, supported 
along by the computational model in \cite{fgiv-cstmabci-12}, that the minimal ABC tree is unique
for trees of order larger than $168$.

%\begin{conjecture} \label{conj-no-pendant-paths}
%A minimal-ABC tree of order $n > 168$ is unique.
%\end{conjecture}

%\begin{conjecture} \label{conj-no-degree-6}
%Except the root vertex, no other vertex of a minimal-ABC tree can have degree larger than  $5$.
%\end{conjecture}

% ----------------------------------------------------------------

%  ---------------------------------------------------------------
\section[Acknowledgment]{Acknowledgment}
% ------------------------------------------------------------------
The author thanks Boris Furtula for
presenting the problem of characterizing graphs with minimal
atom-bond connectivity index and sharing 
initial information about it.
%
%
% -------------------------------------------------------------------------
% ----------------------           bibliography       ---------------------
% -------------------------------------------------------------------------
%

%

\end{document}